\pdfoutput=1

\PassOptionsToPackage{dvipsnames}{xcolor} 
\documentclass{lipics-v2021}

\usepackage[dvipsnames]{xcolor}
\usepackage{amsmath, amssymb, amsthm}   
\usepackage{graphics}
\usepackage{wrapfig}
\usepackage{caption}
\usepackage{verbatim}
\usepackage{cancel}
\usepackage{xspace}
\usepackage{multicol}
\usepackage{subcaption}

\definecolor{navy}{RGB}{0,0,128}

\newcommand{\relu}{\text{ReLU}\xspace{}}

\usepackage{tikz,ifthen,pgfplots}
\usetikzlibrary{arrows,trees,backgrounds,automata,shapes,decorations,plotmarks,fit,calc,positioning,shadows,chains}
\tikzstyle{every pin edge}=[<-,shorten <=1pt]
\tikzstyle{neuron}=[circle,fill=black!25,minimum size=17pt,inner sep=0pt]
\tikzstyle{input neuron}=[neuron, fill=green!40]
\tikzstyle{output neuron}=[neuron, fill=red!40]
\tikzstyle{hidden neuron}=[neuron, fill=blue!40]
\tikzstyle{constructed neuron new}=[neuron, color=red, fill=orange!40]
\tikzstyle{constructed neuron}=[neuron, fill=orange!40]
\tikzstyle{annot} = [text width=6em, text centered]
\tikzstyle{nnedge} = [-{stealth},shorten >=0.1cm, shorten <=0.05cm,line width=0.8pt,black]
\tikzstyle{proofNode} = [rounded rectangle, fill=purple!30]
\tikzstyle{proofEdge} = [-{stealth},shorten >=0.1cm, shorten <=0.05cm,line width=0.8pt,black]
\usetikzlibrary{calc}

\newcommand{\mysubsection}[1]{\medskip\noindent\textbf{#1}}

\newcommand{\sat}{\texttt{SAT}}
\newcommand{\unsat}{\texttt{UNSAT}}

\newcommand{\nn}{\mathcal{N}}
\newcommand{\rn}[1]{\mathbb{R}^{#1}}
\newcommand{\tone}{Th(\mathbb{R},+,-,\cdot,0,1,<)}
\newcommand{\ttwo}{Th(\mathbb{R},+,-,\cdot,0,1,<,\exp)}
\newcommand{\tthree}{Th(\mathbb{R},+,-,{\cdot}_{q\in\mathbb{Q}},0,1,<,\sigma,\tanh,\tau)}
\newcommand{\tfour}{Th(\mathbb{R},+,-,\cdot,0,1,<,e)}

\newcommand{\definabletwo}[2]{ (\exp(#1) + 1) \cdot #2 = \exp(#1)}
\newcommand{\definablethree}[2]{ (\exp(#1) + \exp(-#1)) \cdot #2 = \exp(#1) - \exp(-#1)}

\title{DNN Verification, Reachability, and the Exponential Function Problem}
\author{Omri Isac}{The Hebrew University of Jerusalem}{}{}{}
\author{Yoni Zohar}{Bar-Ilan University}{}{}{}
\author{Clark Barrett}{Stanford University}{}{}{}
\author{Guy Katz}{The Hebrew University of Jerusalem}{}{}{}
\authorrunning{O. Isac, Y. Zohar, C. Barrett and G. Katz}

\Copyright{Omri Isac, Yoni Zohar, Clark Barrett and Guy Katz}
\hideLIPIcs
\nolinenumbers

\begin{document}

\maketitle

\begin{abstract}
  Deep neural networks (DNNs) are increasingly being deployed to
  perform safety-critical tasks. The opacity of DNNs,
  which prevents humans from reasoning about them, presents new
  safety and security challenges. To address these challenges, the
  verification community has begun developing techniques for
  rigorously analyzing DNNs, with numerous verification algorithms
   proposed in recent years.  While a significant amount of work
  has gone into developing these verification algorithms, little work
  has been devoted to rigorously studying the computability and complexity of the
  underlying theoretical problems.  Here, we seek to contribute to the
  bridging of this gap.  We focus on two kinds of DNNs: those that
  employ piecewise-linear activation functions (e.g., \relu{}), and those that
  employ piecewise-smooth activation functions (e.g., Sigmoids).  We prove the
  two following theorems:
\begin{enumerate}
\item The decidability of verifying DNNs with a particular set of piecewise-smooth activation functions, including Sigmoid and $\tanh$,
   is equivalent to a well-known, open problem formulated by
  Tarski; and
\item The DNN verification problem for any quantifier-free linear
  arithmetic specification can be reduced to the DNN reachability
  problem, whose approximation is NP-complete.
\end{enumerate}
These results answer two fundamental questions about the
computability and complexity of DNN verification, and the ways it is
affected by the network's activation functions and error tolerance;
and could help guide future efforts in developing DNN verification
tools.
\end{abstract}

\section{Introduction}
\label{sec:Introduction}

The use of artificial intelligence, and specifically that of deep
neural networks (DNNs), is becoming extremely widespread --- as DNNs
are often able to solve complex tasks more successfully than any other
computational approach. These include critical tasks in
healthcare~\cite{EsRoRaKuDeChCuCoThDe19}, autonomous
driving~\cite{BoDeDwFiFlGoJaMoMuZhZhZhZi16},
communication networks~\cite{ChChSaYiDe19}, and also the task of communicating
with humans through text~\cite{chatGPT22} --- which seems to bring
DNNs closer and closer to passing the famous Turing test~\cite{Tu50}.

However, it has been shown that even state-of-the-art DNNs are
susceptible to various errors. In one infamous example, known as
\emph{adversarial perturbations}, small input perturbations that are
imperceivable to the human eye are crafted to fool modern DNNs,
causing them to output incorrect results selected by an
attacker~\cite{GoShSz14}. Adversarial perturbations thus constitute a
safety and security threat, to which most DNNs are
susceptible~\cite{SzZaSuBrErGoFe13}. Other issues, such as privacy concerns and bias
against various groups, have also been observed, making it clear that
a high bar of trustworthiness must be met before stakeholders can
fully accept DNNs~\cite{LiXiWaZoXiYiVa20}.

Overcoming these weaknesses of DNNs is a significant challenge, due to
their size and complexity. This is further aggravated by the fact that
DNNs are machine-generated (automatically \emph{trained} over many
examples). Consequently, they are  opaque to human engineers, and often fail to generalize their results to examples sufficiently different from the set of examples used for training~\cite{LuScHe21}.  This has
sparked much interest in the verification community, which began
studying verification techniques for DNNs, in order to guarantee their
compliance with given specifications. In recent years, the
verification community has designed and implemented multiple
verification algorithms for DNNs, relying on techniques such as SMT
solving~\cite{HuKwWaWu17,KaBaDiJuKo21,WuZeKaBa22}, abstract
interpretation~\cite{GeMiDrTsChVe18, GuPuTa21}, convex
relaxation~\cite{LyKoKoWoLiDa20}, adversarial search~\cite{HeLo20},
and many others~\cite{ZhShGuGuLeNa20, AvBlChHeKoPr19, BaShShMeSa19,
  FrChMaOsSe20, AkKeLoPi19, PuTa10, SaDuMo19, GaGePuVe19,
  TrBaXiJo20,WaPeWhYaJa18,OsBaKa22, ElKaKaSc21, StWuZeJuKaBaKo21, AmScKa21}.  Indeed, DNN verification technology has
been making great strides recently~\cite{MuBrBaLiJo22}.

Modern verification algorithms depend heavily on the structure of the
DNN being verified, and specifically on the type of its
\emph{activation functions}. Initial efforts at DNN verification
focused almost exclusively on DNNs with piecewise-linear (PWL)
activation functions. It has been shown that the
verification of such networks is an NP-complete problem~\cite{KaBaDiJuKo21,SaLa21}, and multiple
algorithms have been proposed for solving it~\cite{LuMa17,KaBaDiJuKo21,GeMiDrTsChVe18}.  Although more recent
approaches can handle DNNs with smooth activation functions, these
algorithms are often approximate and/or incomplete~\cite{MuMaSiPuVe22,HeLo20,GaGePuVe19}; in fact,
to the best of our knowledge, there is not a single algorithm
that is guaranteed to terminate with a correct answer when verifying
such DNNs. This raises two important
questions: 
\begin{enumerate}
	\item Does there exists a non-approximating algorithm that can
\emph{always} solve verification queries involving DNNs with non-PWL
activation functions?  Or, in other words, is the verification problem
of DNNs with smooth and piecewise-smooth activation functions \emph{decidable}? and \item When introducing approximations, how difficult does the verification problem become with respect to the DNN, the specification, and the size of the approximation? In other words, what is the computational complexity of DNN verification with smooth and piecewise-smooth activation functions and with $\epsilon$-error tolerance? 
\end{enumerate}

In this paper, we provide a partial answer for the first question, by
showing that the verification problem of DNNs with smooth and
piecewise-smooth activation functions is equivalent to a well-known,
open problem from the field of model theory --- Tarski's exponential
function problem~\cite{Ta52}. We do so by introducing a constructive
bijection between verification queries of such DNNs, and instances of
Tarski's open problem.

In addition, we provide a partial answer to the second question, by studying
the relations between DNN verification and DNN reachability problems,
and ultimately proving that they are equivalent. Even though this equivalence result was previously used~\cite{ElGoKa20}, as far as we know, we are the first to provide a formal reduction. This enables further
investigation of DNN verification with any
quantifier-free linear arithmetic specification formula as a specific case of DNN
reachability, without loss of generality. The latter problem is known
to be NP-complete when $\epsilon$-error tolerance is introduced in the result~\cite{RuHuZiKw18}.

Formally, we prove the two following theorems:
\begin{enumerate}
	\item The DNN verification problem for DNNs with smooth and piecewise-smooth activation functions is equivalent
          to Tarski's exponential function problem~\cite{Ta52}, which
          is a well-known open problem. 
	\item The DNN verification problem, with any
          quantifier-free linear arithmetic specification formula, can
          be reduced to the DNN reachability problem, which is
          NP-complete when some $\epsilon$-error tolerance is allowed~\cite{RuHuZiKw18}.
\end{enumerate}
Our results imply a fundamental difference between the hardness of
verification of DNNs with piecewise-smooth and with piecewise-linear
activation functions. As far as we know, we are the first to provide
any proof of this difference.

The rest of the paper is organized as follows. In Section~\ref{sec:Background} we provide background on DNNs,
verification and other necessary mathematical concepts. In Section~\ref{sec:dec} and Section~\ref{sec:verisreach} we formally prove the two main results
mentioned above. In Section~\ref{sec:Related} we discuss related work,
and in Section~\ref{sec:Conclusion} we describe our conclusions and
directions for future work.

\section{Background}
\label{sec:Background}
\subsection{Deep Neural Networks}
\textit{Deep neural networks}
(DNNs)~\cite{GoBeCu16} are directed graphs whose nodes (neurons) are
organized into layers, and whose nodes and edges are labeled with rational numbers. Nodes in the first layer, called the
\textit{input layer}, are assigned values matching the input to the DNN; and then the
values of nodes in each of the subsequent layers are computed
as functions of the values assigned to neurons in the preceding layer.  More
specifically, each node value is computed by first applying an affine 
transformation (linear transformation and addition of a constant) to the values from the preceding layer, and then
applying a non-linear \textit{activation function}~\cite{DuSiCh22} to the result. The
final (output) layer, which corresponds to the output of the network,
is computed without applying an activation function.

Three of the most common activation functions are the
\textit{rectified linear unit} (\relu), which is defined as:
\begin{center}
	$  \relu(x) = \begin{cases}
		x & x>0 \\
		0 & \text{otherwise;} \end{cases} $
\end{center}
the \emph{Sigmoid} function, defined as: 
\begin{center}
	$\sigma(x):\mathbb{R}\rightarrow(0,1) = \frac {\exp(x)}{\exp(x)+1}$
\end{center}
where $\exp$ is the exponential function; and the \emph{hyperbolic tangent}, defined as:
\begin{center}
	$\tanh(x):\mathbb{R}\rightarrow(-1,1) = \frac{\exp(x)-\exp(-x)}{\exp(x)+\exp(-x)}$
\end{center}
The latter two activation functions are both injective, and their inverses are defined as follows:
\begin{center}
	$\forall x\in(0,1): \sigma^{-1}(x) = \ln(\frac{x}{1-x})$ \\
	$\forall x\in(-1,1): \tanh^{-1}(x) = \frac{1}{2} \ln(\frac{1+x}{1-x})$
\end{center}
where $\ln$ is the natural logarithm function.
In addition, we consider the $\text{NLReLU}$ activation
function~\cite{LiZhGaQuJi19,DuSiCh22}, denoted $\tau$ for short, which
is defined as:
\begin{center}
	$\text{NLReLU}(x) = \tau(x) = \ln(\relu(x)+1)$
\end{center}

A simple DNN with four layers appears in Figure~\ref{fig:toyDnn}, where all biases are set to zero and are ignored. For input $\langle 1,2,1\rangle$, the first node in the second layer
evaluates to
$ \relu(1 \cdot 1 \; + \; 2 \cdot (-1) ) = \relu(-1) = 0 $; 
the second node in the second layer evaluates to
$ \relu(2 \cdot 1 \; + \; 1 \cdot (-1) ) = \relu(1) = 1 $;  and the node
in the third layer evaluates to $ \sigma(0-1) = \sigma(-1) $. Thus, the
node in the fourth (output) layer evaluates to $4\cdot\sigma(-1)$.

\begin{figure}[h]
	\begin{center}
		\def\layersep{1.2cm}
		\begin{tikzpicture}[shorten >=1pt,->,draw=black!50, node distance=\layersep,font=\footnotesize]
			
			\node[input neuron] (I-1) at (0,-1) {$x_1$};
			\node[input neuron] (I-2) at (0,-3) {$x_2$};
			\node[input neuron] (I-3) at (0,-5) {$x_3$};
			
			\node[hidden neuron] (H-1) at (\layersep, -2) {$v_1$};
			\node[hidden neuron] (H-2) at (\layersep,-4) {$v_2$};
			
			\node[hidden neuron] (H-3) at (2*\layersep,-3) {$v_3$};
			
			\node[output neuron] (O-1) at (3*\layersep,-3) {$y$};
			
			\draw[nnedge] (I-1) --node[above,pos=0.4] {$1$} (H-1);
			\draw[nnedge] (I-2) --node[below,pos=0.4] {$-1$} (H-1);
			
			\draw[nnedge] (I-2) --node[above,pos=0.4] {$1$} (H-2);
			\draw[nnedge] (I-3) --node[above,pos=0.4] {$-1$} (H-2);
			
			\draw[nnedge] (H-1) --node[below] {$1$} (H-3);
			\draw[nnedge] (H-2) --node[below] {$-1$} (H-3);
			
			\draw[nnedge] (H-3) --node[below] {$4$} (O-1);
			
			\node[above=0.05cm of H-1] (b1) {$\relu{}$};
			\node[above=0.05cm of H-2] (b1) {$\relu{}$};
			\node[above=0.05cm of H-3] (b1) {$\sigma{}$};
			
		\end{tikzpicture}
	\end{center}
	\caption{A toy DNN.}
	\label{fig:toyDnn}
\end{figure}

Formally, a DNN $ \mathcal{N}:\rn{m}\rightarrow\rn{k} $ is a sequence
of $ n $ layers $ L_0,...,L_{n-1} $ where each layer $ L_i $ consists of
$ s_i \in \mathbb{N} $ nodes, denoted $ v^1_i,...,v^{s_i}_i $, and biases $p^j_i\in\mathbb{Q}$ for each $v^j_i$. Each directed edge in the DNN is of the form $(v^l_{i-1},v^j_i)$ and is labeled with $w_{i,j,l}\in\mathbb{Q}$.
The assignment to the nodes in the input layer is defined by $v^j_0 = x_j$, where $\overline{x}\in\rn{m}$ is the input vector, and the assignment for the
$ j^{th} $ node in the $ 1 \leq i < n-1 $ layer is computed as
\begin{center}
	$v^j_i = f^j_i \left( \underset{l=1} { \overset {s_{i-1}} { \sum } }
	w_{i,j,l} \cdot v^l_{i-1} + p^j_i \right)$
\end{center}
for some activation function
$f^j_i:\mathbb{R}\rightarrow\mathbb{R}$.
Finally, neurons in the output layer are computed as: 
\begin{center}
	$v^j_{n-1} = \underset{l=1} { \overset {s_{n-2}} { \sum } }
	w_{n-1,j,l} \cdot v^l_{n-2} + p^j_{n-1}$ 
\end{center}
where
$ w_{i,j,l}$ and $ p^j_i $ are (respectively) the predetermined weights and biases of $ \mathcal{N}$. The size of a network $|\mathcal{N}|$ is defined as the overall number of its neurons.

\subsection{Formal Analysis of DNNs}
\label{subsec:DNNanalaysisBackground}
The formal methods community has tackled the formal analysis of DNNs
primarily along two axes: DNN \emph{verification}, and DNN
\emph{reachability}. These are two related formulations, as
reachability problems may be expressed as verification problems in a
straightforward manner. In this paper, we further study the
connections between these two formulations.

\mysubsection{DNN Verification.}
  Let 
$\nn:\rn{m}\rightarrow\rn{k} $ be a DNN and 
$ P:\rn{m+k}\rightarrow \lbrace {\top,\bot} \rbrace $ be a property, where $\top,\bot$ represent the values for which the property does and does not hold, respectively. The \emph{DNN verification problem} is to decide whether there exist
$ x\in\rn{m}$ and $y\in\rn{k}$ such that
$ ( \mathcal{N}(x) = y )\land P(x,y) $ holds. In particular, a verification query is always expressed as an existential formula.
 If such $x$ and $y$ exist,
we say that the verification query $\langle \mathcal{N},P \rangle$ is
\textit{satisfiable} (\sat); and otherwise, we say that it is
\textit{unsatisfiable} (\unsat).

A verification algorithm is \emph{sound} if it does not return
$\unsat{}$ for  satisfiable queries, and does not return $\sat$ for
unsatisfiable queries (in other words, if its answers are always
correct); and is \emph{complete} if it always terminates, for any query.

So far, there have been several efforts at studying the
complexity-theoretical aspects of DNN
verification~\cite{FeSh21,HeLeZi21,SaLa21,SaLa23,RuHuZiKw18,IvWeAlPaLe19,
  KaBaDiJuKo21}. Most previous work was focused on DNNs with
piecewise-linear activation function (specifically, $\relu$s), while
leaving many open questions about the complexity-theoretical aspects
of verifying DNNs with smooth and piecewise-smooth activation
functions.

\mysubsection{DNN Reachability.} Given a DNN
$ \mathcal{N}:\rn{m}\rightarrow\rn{} $, a function $o:\rn{}\rightarrow\rn{}$ and an input
set $\mathcal{X}\subseteq[0,1]^m$, the \emph{DNN reachability problem}
is to compute $\underset{x\in\mathcal{X}}{\sup}\ o(\mathcal{N}(x))$
and $\underset{x\in\mathcal{X}}{\inf}o(\mathcal{N}(x))$, perhaps up
to some $\epsilon$-error tolerance. For DNNs with
Lipschitz-continuous activation functions, either smooth or
piecewise-linear, (such as $\sigma$ and \relu{}), the reachability
problem with some $\epsilon$-error tolerance is
NP-complete, in the size of the network and $\epsilon$~\cite{RuHuZiKw18}. 
In this work, we consider a decision version for the problem where $o$ is the identity,
deciding whether $\underset{x\in\mathcal{X}}{sup}\ \mathcal{N}(x) \geq 0$ is
achieved for some $x\in\mathcal{X}$. This decision version with some $\epsilon$-error tolerance is then to decide whether $\underset{x\in\mathcal{X}}{sup}\ \mathcal{N}(x) \geq -\epsilon$ is achieved for some $x\in\mathcal{X}$.  In addition, we may assume that
the input of $\mathcal{N}$ is within $[0,1]^m$, as the input domain
may be normalized before the network is evaluated.

For example, consider the DNN depicted in Figure~\ref{fig:toyDnn}. A possible
\textit{verification} query for this DNN is given by a property
$P$ that returns $\top$ if and only if
$(x_1,x_2,x_3)\in [0,1]^3 \wedge (y\in [0.5, 0.75] \lor y\in [0,
0.25])$; i.e., if there exists an input in the $[0,1]^3$ cube, for
which $y\in [0.5, 0.75] \lor y\in [0, 0.25]$. A possible \textit{reachability}
query is to check whether there exists an input in
the domain $[0,1]\times[0,0.5]\times[0.5,1]$, for which $y \geq
0$. This reachability query can trivially be represented as a verification property $P'$,
which returns $\top$ if and only if
$(x_1,x_2,x_3)\in [0,1]\times[0,0.5]\times[0.5,1] \wedge y\geq 0$. For
any $\epsilon>0$, the equivalent reachability query with $\epsilon$
tolerance is to decide if there exists an input in the domain
$[0,1]\times[0,0.5]\times[0.5,1]$, for which $y \geq -\epsilon $.

\subsection{Decidability and Mathematical Logic}

\mysubsection{Mathematical Logic.}  In mathematical logic, a
\emph{signature} $\Sigma$ is a set of symbols, representing functions
and relations.  A \emph{$\Sigma$-formula} is a formula, comprised of
atoms and relations that appear in $\Sigma$, the usual logical
operators $(\wedge, \lnot, \lor, \rightarrow, \leftrightarrow)$, and
the quantifiers $\forall$ (universal) and $\exists$ (existential). A variable affixed with a
quantification symbol is a \emph{bounded variable}; and otherwise it
is a \emph{free variable}. A formula without free variables is called
a \emph{sentence}, and a formula without bounded variables is called a \emph{quantifier-free formula}. A formula with variables $\overline{x}=(x_1,...,x_n)$ of the form $\exists(\overline{x}).\varphi$ where $\varphi$ is quantifier-free is called an \emph{existential formula}. 
A \emph{$\Sigma$-theory} is a set of $\Sigma$-sentences. A
\emph{$\Sigma$-model} $\mathcal{M}$ is comprised of a set of elements,
denoted $|\mathcal{M}|$, and an \emph{interpretation} for all $\Sigma$
functions and relations; that is, a definition
$f^{\mathcal{M}}:|\mathcal{M}|^n\rightarrow|\mathcal{M}|$ for every
n-ary function $f\in\Sigma$, and a definition
$r^{\mathcal{M}}\subseteq |\mathcal{M}|^m$ or every m-ary relation
$r\in\Sigma$. If the interpretation of a $\Sigma$-sentence $\varphi$
is true within a model $\mathcal{M}$, we say that $\mathcal{M}$
\emph{satisfies} $\varphi$, and denote $\mathcal{M}\models\varphi$.
If $\mathcal{M}$ satisfies all sentences in a $\Sigma$-theory
$\mathcal{T}$, then $\mathcal{M}$ is a \emph{$\mathcal{T}$-model},
denoted $\mathcal{M}\models\mathcal{T}$.  Given some model
$\mathcal{M}$ over signature $\Sigma$, we define the theory
$Th(\mathcal{M})$ as the set of all $\Sigma$-sentences $\varphi$ such
that $\mathcal{M}\models \varphi$. It is then trivial that
$\mathcal{M}\models Th(\mathcal{M})$.

For example, let $\Sigma$ be the set $\lbrace +,-,\cdot,0,1,< \rbrace$ where $+,-,\cdot$ are 2-ary functions, $0,1$ are 0-ary functions and $<$ is a 2-ary relation. Let $\mathcal{M}$ be a model defined over $\mathbb{R}$ with addition, subtraction, multiplication, the constants 0,1 and the usual order. Then $Th(\mathcal{M})$ is the set of all $\Sigma$-sentences that $\mathcal{M}$ satisfies, such as $\forall x: x\cdot0=0$.

\mysubsection{Theory Decidability.}  For a $\Sigma$-theory
$\mathcal{T}$ and a $\Sigma$-sentence $\varphi$, we say that $\varphi$
is \emph{$\mathcal{T}$-valid} and denote $\mathcal{T} \vdash \varphi$,
if every model of $\mathcal{T}$ satisfies $\varphi$.  Furthermore, we
say that $\varphi$ is \emph{$\mathcal{T}$-satisfiable} if there exists
a model $\mathcal{M}$ of $\mathcal{T}$, for which
$\mathcal{M}\models\varphi$; and that $\varphi$ is
\emph{$\mathcal{T}$-unsatisfiable} if $\mathcal{M}\not\models\varphi$
for all models $\mathcal{M}$ of $\mathcal{T}$.  Satisfiability and
validity are closely connected, as $\varphi$ is $\mathcal{T}$-valid if
and only if $\lnot\varphi$ is $\mathcal{T}$-unsatisfiable.  A theory
$\mathcal{T}$ is \emph{decidable} if there exists an algorithm that,
for any sentence $\varphi$, decides whether
$\mathcal{T} \vdash \varphi$, within a finite number of steps. If
$\varphi$ is valid, the algorithm returns $\top$; and otherwise, it
returns $\bot$. Due to the connection of satisfiability and validity,
validity-checking algorithms may also be used to decide
satisfiability, and vice versa. In particular, for any theory
$Th(\mathcal{M})$ for some model $\mathcal{M}$, all
$Th(\mathcal{M})$-models satisfy exactly the same sentences, so
validity and satisfiability are equivalent. Thus, throughout this
paper we use decision procedures to decide the
$Th(\mathcal{M})$-satisfiability of formulas.  In addition, when
considering \emph{quantifier-free} formulas (i.e., a formula where all
variables are free), all of the formula's variables are implicitly
existentially quantified. In this case, for any quantifier-free
formula $\varphi(\overline{x})$ with variable vector $\overline{x}$,
the satisfiability problem of $\varphi(\overline{x})$ with respect to
a model $\mathcal{M}$ is equivalent to deciding whether
$\mathcal{M}\models\exists\overline{x}.\varphi(\overline{x})$. Similarly, the satisfiability problem of $\varphi(\overline{x})$ with respect to
a theory $\mathcal{T}$ is equivalent to deciding whether
$\exists\overline{x}.\varphi(\overline{x})$ is $\mathcal{T}$-satisfiable.

It has previously been shown that the theory of the real field $\tone$
is decidable~\cite{Ta48}, and that the theory of the real field with
the transcendental functions $\exp, \sin$ and the constants
$\log2, \pi$ is undecidable~\cite{Ri69}. The question of the decidability of
$\ttwo$ has remained an open problem since the
1950's~\cite{Ta52}, and it is commonly known as Tarksi's
\emph{exponential function problem}.

A theory $\mathcal{T}$ is \emph{stably-infinite} if for every quantifier-free formula $\varphi$, the satisfiability of $\varphi$ in $\mathcal{T}$ implies that $\varphi$ is satisfiable in some infinite model of $\mathcal{T}$. It is then immediate that for any infinite model $\mathcal{M}$, $Th(\mathcal{M})$ is stably-infinite. 
 
Given two decidable, stably-infinite theories
$\mathcal{T}_1$ and $\mathcal{T}_2$ defined over disjoint sets of symbols,
for any quantifier-free formulas
$F_1\in \mathcal{T}_1, F_2\in \mathcal{T}_2$, the formula
$F_1\wedge F_2$ is decidable as well~\cite{NeOp79}.  The
\emph{Nelson-Oppen method}~\cite{NeOp79} is a well-known method for combining two
decision procedures for two theories into a decision procedure for the
quantifier-free fragment of their union.

\mysubsection{Equisatisfiability.} Two formulas $\varphi$ and $\psi$ are \emph{equisatisfiable} if $\varphi$ is satisfiable if and only if $\psi$ is satisfiable.
For example, the formulas $\varphi := (a+b)*(a-b)=0$ and $\psi :=
c*d=0 \wedge c = a+b \wedge d = a-b$ are equisatisfiable. Note that
$\varphi$ and $\psi$ may be formulated in different theories,
$\mathcal{T}_1$ and $\mathcal{T}_2$, respectively. In this case, we say that the formulas are equisatisfiable if $\varphi$ is $\mathcal{T}_1$-satisfiable if and only if $\psi$ is $\mathcal{T}_2$-satisfiable.

\mysubsection{Function Definability.}
For any signature $\Sigma$ and an n-ary function $f$, not necessarily
in $\Sigma$, we say that $f$ is \emph{definable} in a $\Sigma$-model
$\mathcal{M}$ if there exists a $\Sigma$-formula
$\psi(x_1,...,x_n,y,z_1,...,z_m)$ over the variables $x_1,...,x_n,y$
such that for any elements $a_1,...,a_n,b$ in $\mathcal{M}$ we have that
$\mathcal{M}\models\exists z_1,...,z_m.\psi(a_1,...,a_n,b,z_1,...,z_m)$ if and only if $b = f(a_1,...,a_n)$. We say that $f$ is definable in a $\Sigma$-theory $\mathcal{T}$ if it is definable in all models of $\mathcal{T}$.

\mysubsection{Model-Completeness.} In model theory, the concept of model-completeness has several equivalent definitions. For our purposes, a theory $\mathcal{T}$ is model-complete if and only if any formula in the theory has an equivalent existential  formula (modulo $\mathcal{T}$). This means that the existential formulas in $\mathcal{T}$ can express all the formulas in it. $\ttwo$ is known to be model-complete~\cite{Wi96}.

\section{Decidability of DNN Verification}
\label{sec:dec}
In this section, we prove our first main result: the decidability of verifying
a DNN with the activation functions \relu{}, $\sigma$, $\tanh$ and $\tau$ is equivalent to the decidability of \linebreak $\ttwo$. The decidability of this theory is an open problem~\cite{Ta52}. Thus, the equivalence implies that the decidability of DNN
verification for DNNs with the activation functions \relu{}, $\sigma$, $\tanh$ and $\tau$ is an open problem as well.

For simplicity, we denote $\mathcal{T}_{\mathbb{R}} = \tone$,
$\mathcal{T}_{\exp}=\ttwo$, and $\mathcal{T}_{\sigma}=\tthree$, where
$\cdot_q$ is an unary function, interpreted as the multiplication with
a constant $q\in\mathbb{Q}$.  We use $\Sigma_{\sigma},\Sigma_{\exp}$
and $\Sigma_{\mathbb{R}}$ to denote the signatures of
$\mathcal{T}_{\sigma},\mathcal{T}_{\exp}$ and
$\mathcal{T}_{\mathbb{R}}$, respectively.  Note that for any DNN,
weights and biases are in $\mathbb{Q}$, and can thus be expressed as
$\mathcal{T}_{\mathbb{R}}$-terms. Therefore, we can express the affine
constraints of the network as
$\mathcal{T}_{\mathbb{R}}$-formulas. In addition, any constraint of the
form $f=\relu(b)$ can be expressed as the formula
$(f=b \leftrightarrow b > 0) \wedge (f=0 \leftrightarrow b \leq 0) $,
and thus \relu{} is definable in $\mathcal{T}_{\mathbb{R}}, \mathcal{T}_{\exp}$ and
$\mathcal{T}_{\sigma}$. Therefore, without loss of generality, we need
not add a function symbol to $\Sigma_{\sigma}$ to express DNNs with
\relu{} activation functions (or any other piecewise-linear function). Our
goal is then to show that the decidability of $\mathcal{T}_{\exp}$ is
equivalent to the decidability of all existential formulas of
$\mathcal{T}_{\sigma}$.

\mysubsection{Example:} We begin with an example that illustrates this
equivalence. For the first direction, consider the toy DNN depicted
in Figure~\ref{fig:toyDnn}, and let $x_1,x_2,x_3$,
$b_1,f_1,b_2,f_2,b_3,f_3$, and $y$ be the variables of the
network. Variables $x_1,x_2,x_3$ represent the input variables,
variables $b_1,f_1,b_2,f_2,b_3,f_3$ represent the inputs and outputs
of nodes $v_1,v_2,v_3$, respectively, and variable $y$ represents the network's output. Let $P$ be the property restricting the input to be within $[0,1]^3$ and the output to be in $[1,2]$. The verification query for the network in Figure~\ref{fig:toyDnn} and $P$ is then:

\begin{center}
	$\underset{i\in{1,2,3}}{\bigwedge}[(x_i \geq 0) \wedge (x_i \leq 1)] \wedge$ \\
	$(x_1 - x_2 = b_1) \wedge (x_2 - x_3 = b_2) \wedge$ \\
	$\underset{i\in{1,2}}{\bigwedge}[(f_i=b_i) \leftrightarrow (b_i > 0)] \wedge [(f_i=0) \leftrightarrow (b_i \leq 0)] \wedge$ \\
	$(f_1 - f_2 = b_3) \wedge (f_3 = \sigma(b_3)) \wedge (4\cdot f_3 = y) \wedge $  \\
	$(1\leq y) \wedge(y \leq 2)$\\
\end{center}
This is a $\mathcal{T}_{\sigma}$ query, which can be expressed as a
query in $\mathcal{T}_{\exp}$, since $\sigma(x) =
\frac{\exp(x)}{1+\exp(x)}$. The equivalent $\mathcal{T}_{\exp}$ query
is:

\begin{center}
	$\underset{i\in{1,2,3}}{\bigwedge}[(x_i \geq 0) \wedge (x_i \leq 1)] \wedge$ \\
	$(x_1 - x_2 = b_1) \wedge (x_2 - x_3 = b_2) \wedge$ \\
	$\underset{i\in{1,2}}{\bigwedge}[(f_i=b_i) \leftrightarrow (b_i > 0)] \wedge [(f_i=0) \leftrightarrow (b_i \leq 0)] \wedge$ \\
	$(f_1 - f_2 = b_3) \wedge [\definabletwo{b_3}{f_3}] \wedge (4\cdot f_3 = y)$ \\
	$(1\leq y) \wedge(y\leq2)$\\
\end{center}

For the second direction, we begin by demonstrating a purification
process of a given formula $\varphi := \exp(a+b)=\exp(a)\cdot\exp(b)$
into a formula in $\mathcal{T}_\sigma$. We assume that we can define
$\psi_{c=a\cdot b}$ and $\psi_{y=\exp(x)}$ in $\Sigma_\sigma$, which
are defined over the variables $a,b,c$ and $x,y$, respectively, and that witness the definability of the functions $\cdot$ and $\exp$ in $\mathcal{T}_{\sigma}$.
 Therefore, the formula
\begin{center}
		$\psi_{p=\exp(a)} \wedge \psi_{q=\exp(b)} \wedge \psi_{r=\exp(a+b)} \wedge
		\psi_{r=p\cdot q} $
\end{center}
is equisatisfiable to $ \exp(a+b)=\exp(a)\cdot\exp(b)$.

Formally, we prove the following theorem:
\begin{theorem}
	The decidability of verifying DNNs with $\sigma, \tanh, \tau$
        and $\relu{}$  activation functions is equivalent to the decidability of $\ttwo$.
     \label{thm:sigmoids}
\end{theorem}

\begin{proof}
  The first direction of the proof is similar to a technique proposed
  by Ivanov et al.~\cite{IvWeAlPaLe19}. Assume there exists a
  decision procedure for $\mathcal{T}_{\exp}$, and let
  $F\in\Sigma_{\sigma}$ be a DNN verification
  query. For any appearance of $\sigma(t)$ for some term
  $t$, we replace $t,\sigma(t)$ with the fresh variables $x,y$,
  respectively and add the conjunction: 
  \begin{center}
  $\definabletwo{x}{y} \wedge x = t$
  \end{center}
  to the resulting formula. This is done in a way similar to the one
  described in the example.
  Similarly, for any appearance of $\tanh(t)$ for some term
  $t$, we replace $t,\tanh(t)$ with the fresh variables $x,y$,
  respectively and add the conjunction: 
    \begin{center}
    	$\definablethree{x}{y} \wedge x = t$
   \end{center}
  to the resulting formula.
  For defining $\tau$, we first define: 
    \begin{center}
    	$\psi_{f=\relu(b)}:= (f=b \leftrightarrow b > 0) \wedge (f=0 \leftrightarrow b \leq 0) $
    \end{center}
  Now, for any appearance of $\tau(t)=\ln(\relu(t)+1)$ for some term
  $t$, we replace $t,\tau(t)$ with the fresh variables $x,y$, respectively and add the conjunction:
  \begin{center}
  	$ \psi_{z=\relu(x)} \wedge \exp(y) = z+1 \wedge x = t$
  \end{center}
  to the resulting formula, where $z$ is an additional fresh variable.
  After repeating this process iteratively,
  we convert any $F\in\Sigma_{\sigma}$ to an equisatisfiable
  formula $F'\in\Sigma_{\exp}$. We then use the decision
  procedure to decide the satisfiability of $F'$.

  The second direction of the proof is more complex. Assume we have a
  sound and complete verification procedure for DNNs with $\sigma,
  \tanh$ and $\tau$ activation functions; that is,
  a decision procedure for deciding the satisfiability of quantifier-free $\Sigma_\sigma$-formulas in $\mathcal{T}_{\sigma}$.

  Since $\mathcal{T}_{exp}$ is model-complete, it is tempting to try and
  construct a decision procedure for the existential formulas of
  $\mathcal{T}_{\exp}$. However, to the best of our knowledge, given a
  general formula in $\mathcal{T}_{\exp}$ it is not known how to
  effectively derive its equivalent existential formula. In order to circumvent
  this issue, we consider instead a fourth theory, $\mathcal{T}_e = \tfour$, defined over the signature $\Sigma_e$, where $e:\mathbb{R}\rightarrow\mathbb{R}$ with
  $e(x)=\exp(\frac{1}{1+x^2})$ is the \emph{restricted} exponential
  function.  It has been shown by Macintyre and Wilkie~\cite{MaWi96} that the decidability of this theory implies the decidability of $\mathcal{T}_{exp}$, and
  that given any formula in the language of $\mathcal{T}_e$, one can
  effectively find an equivalent existential formula (in
  $\mathcal{T}_{e}$). Therefore, it is enough for our
  purpose to consider any existential formula $\exists\overline{x}.\varphi\in\Sigma_e$, and decide the satisfiability of $\varphi$ in $\mathcal{T}_e$.

  Let $\exists\overline{x}.\varphi\in\Sigma_e$ be an existential
  formula, where $\varphi$ is a quantifier-free formula. We construct
  a $\Sigma_\sigma$-formula $\psi$, equisatisfiable to $\varphi$. In
  this construction, all variables are implicitly existentially quantified.
  In order to do so, it is enough to define formulas $\psi_{c=a\cdot
    b}$ and $\psi_{y=e(x)}$ over the variables $a,b,c$ and $x,y$
  respectively, and witness the definability of the functions $\cdot$
  and $e$ in $\mathcal{T}_{\sigma}$. In this case, given any formula
  $\varphi\in\Sigma_{e}$, we can iteratively replace any occurrence of
  terms of the form $t\cdot s$ with the fresh variable $p$ and add the
  conjunction $\psi_{p=t\cdot s}$, and occurrences of terms of the
  form $e(x)$ with the fresh variable $q$ and add the conjunction
  $\psi_{q=e(x)}$. This process terminates with a
  $\Sigma_{\sigma}$-formula $\psi$ equisatisfiable to $\varphi$,
  allowing us to apply the decision procedure to $\psi$.
  
  To complete the proof, it remains to show how $\psi_{c=a\cdot b}$ and
  $\psi_{y=e(x)}$ can be defined using the formula $\psi_{y=ln(x)}$.
  We show the construction of $\psi_{y=ln(x)}$ later, in Lemma~\ref{lemma:psiln}, and we use it here to define both $\psi_{c=a\cdot b}$ and $\psi_{y=e(x)}$.
  
  We start by defining $\psi_{c=a\cdot b}$.  Note that
  $\forall a,b > 0$, it holds that $\ln(a\cdot b)=\ln(a)+\ln(b)$; and
  so $a\cdot b = \exp(\ln(a)+\ln(b))$, assuming $a,b > 0$. This
  equality can be expressed using the formula:
  \begin{center}
   $\theta_{c=a\cdot b} := \psi_{p=\ln(a)} \wedge \psi_{q=\ln(b)} \wedge \psi_{p+q=\ln(c)}$,
  \end{center}
  where $c$ represents the value of $a\cdot b$, and $p,q$ are fresh variables.
  For defining $\psi_{c=a\cdot b}$ for all $a,b\in\mathbb{R}$, we split into cases, and write:
  \begin{center}
  	$\psi_{c=a\cdot b} := [(a>0\wedge b>0) \rightarrow \theta_{c=a\cdot b}] \wedge $
  	
  	$[(a<0\wedge b>0) \rightarrow \theta_{-c=-a\cdot b} ] \wedge$
  	
  	$[(a>0\wedge b<0) \rightarrow \theta_{-c=a\cdot -b}] \wedge $
  	
  	$[(a<0\wedge b<0) \rightarrow \theta_{c=-a\cdot-b}] \wedge $
  	
  	$[(a=0\lor b=0) \leftrightarrow c=0]$,
  	
  \end{center}
  which represents the function $\cdot$ and witnesses its definability in $\mathcal{T}_\sigma$.
  
   We now define $\psi_{y=e(x)}$. Recall that $e(x)=\exp(\frac{1}{x^2 +1 } )$, so in order to define $\psi_{y=e(x)}$ we use both $\psi_{c=a\cdot b}$ and $\psi_{y=\ln(x)}$:
   
   \begin{center}
   		$\psi_{y=e(x)} := \psi_{a=\ln(y)} \wedge$
   		$\psi_{1= a\cdot (b + 1)} \wedge $
   		$ \psi_{b=x\cdot x}$
   \end{center}
   where $a,b$ are fresh variables.

We have defined both $\psi_{c=a\cdot b}$ and $\psi_{y=e(x)}$, which concludes our proof.
\end{proof}

For the completeness of this section, we provide now the proof of Lemma~\ref{lemma:psiln}, which shows the construction of $\psi_{y=ln(x)}$:
\begin{lemma}
	\label{lemma:psiln}
	The natural logarithm function $\ln$ is definable in $\mathcal{T}_{\sigma}$.
\end{lemma}
\begin{proof}
	First, observe that for any $x\geq1$ we have that $\tau(x-1) = \ln(\relu{(x-1)}+1)=\ln(x-1+1)=\ln(x)$. 
	Second, observe that $\forall x\in(0,1)$, the inverses of $\sigma$ and $\tanh$ are defined and are equal to:
	\[\sigma^{-1}(x)=\ln(\frac{x}{1-x})= \ln(x)-\ln(1-x)\]
	and 
	\[\tanh^{-1}(x)=\frac{1}{2}\ln(\frac{1+x}{1-x})=\frac{1}{2}(\ln(1+x)-\ln(1-x))\]
	We conclude that:
	\[\forall x\in(0,1): \sigma^{-1}(x) - 2\tanh^{-1}(x) + \tau(x) =  \ln(x) - \ln(1-x) - \ln(1+x) + \ln(1-x) + \ln(x+1)  = \ln(x) \]
	We can express this relation using the following formula, and the fresh variables $a,b,c$:
	\begin{center}
		$ \theta_{x,y} := [x = \sigma(a)] \wedge [x = \tanh(b)] \wedge [c = \tau(x)] \wedge [y = a - 2b + c] $
	\end{center}
	Where $2b$ is syntactic sugar for $\cdot_2(b)$.
	Thus, we can define:
	\begin{center}
		$\psi_{y=\ln(x)} := [(1<x)\rightarrow(y=\tau(x-1))]
		\wedge [(x=1)\leftrightarrow(y=0)]  \wedge  [(0<x<1) \rightarrow  \theta_{x,y}]
		\wedge [0<x]$
              \end{center}
              which concludes the proof.
\end{proof}
\section{DNN Verification is DNN Reachability}
\label{sec:verisreach}

The two main formal analysis approaches for DNNs, verification and
reachability, are closely connected: a DNN reachability instance can
be formulated as DNN verification in a straightforward manner, as in the example in Section~\ref{subsec:DNNanalaysisBackground}. 
Presently, DNN analysis
algorithms and tools typically support one of the two
formulations. Here, we prove that DNN verification and DNN
reachability are in fact equivalent.

In this part, we consider DNNs that use both piecewise-linear and
Sigmoidal activation functions. We formally prove that any instance of
the DNN verification problem, with any specification expressible by a
quantifier-free linear arithmetic formula, can be reduced to an
instance of the DNN reachability problem. Since the reachability
problem is a specific case of verification, we ultimately prove that
for DNNs, reachability and verification are equivalent.  Since it was
shown that approximation of DNN reachability queries with
Lipschitz-continuous activation functions (such as $\sigma$ and
\relu{}) is NP-complete~\cite{RuHuZiKw18}, we deduce that the DNN
verification problem is reducible to a problem whose approximation is NP-complete. 
The reduction involves adding an additional input, denoted $\epsilon$,
and we use $(x,\epsilon)$ to denote the concatenation of $\epsilon$ to the input vector $x$.
Formally, we prove the following theorem:
\begin{theorem}
	Let $\nn: \rn{m} \rightarrow \rn{k}$ be a neural network,
	let $\varphi$ be
	a quantifier-free property with atoms expressing affine constraints over variables $y_i$ of $\nn$, and let $X\subseteq \rn{m}$. There exists a
	neural network $\nn':\rn{m+1} \rightarrow \mathbb{R}$, with
	$|\nn'|=O(|\nn|+|\varphi|)$ such that the two following conditions
	are equivalent:
	\begin{itemize}
		\item
		$\exists x \in X.\quad \nn(x) \models \varphi$
		\item 
		$\exists (x,\epsilon) \in X \times \left(0,1 \right].\quad
		\nn'(x,\epsilon) \geq 0$
	\end{itemize}
\end{theorem}

\mysubsection{Example.}
We begin with an example for constructing $\nn'$, given some DNN
$\nn,$ and a property $\varphi$.
Consider first $\nn: \rn{4} \rightarrow \rn{2}$ as depicted in Figure~\ref{fig:ex1} and $\varphi := (y_1 > 0) \wedge ( y_1 \geq y_2 )$. We denote $\theta := y_1 \geq y_2$ and $\psi := y_1 > 0 \equiv \lnot (-y_1 \geq 0)$.
In Figure~\ref{fig:example}
we start with the initial DNN $\nn$, and then iteratively add new nodes to
$\nn$. In particular,  we show how to add neurons that are active if and only if $\nn \models \theta$ and $\nn \models \psi$, respectively in Figure~\ref{fig:ex2} and Figure~\ref{fig:ex3}. Lastly, in Figure~\ref{fig:ex4} we show how to add the output neuron, such that $\exists x \in X.\quad \nn(x) \models \varphi$ if and only if $\exists (x,\epsilon) \in X \times \left(0,1 \right].\quad \nn'(x,\epsilon) \geq 0$. This concludes our example.

\begin{figure*}[p]
	\begin{subfigure}[t]{0.5\linewidth}
		\def\layersep{1.2cm}

		\begin{center}
			\begin{tikzpicture}[shorten >=1pt,->,draw=black!50, node distance=\layersep,font=\footnotesize]
				
				\node[input neuron] (I-1) at (0,-1) {$x_1$};
				\node[input neuron] (I-2) at (0,-2) {$x_2$};
				\node[input neuron] (I-3) at (0,-3) {$x_3$};
				\node[input neuron] (I-4) at (0,-4) {$x_4$};
				
				\node[hidden neuron] (H-1) at (1.5*\layersep, -1.5) {$v_1$};
				\node[hidden neuron] (H-2) at (1.5*\layersep,-2.5) {$v_2$};
				\node[hidden neuron] (H-3) at (1.5*\layersep, -3.5) {$v_3$};
				
				\node[output neuron] (O-1) at (3*\layersep,-2) {$y_1$};
				\node[output neuron] (O-2) at (3*\layersep,-3) {$y_2$};
				
				\draw[nnedge] (I-1) --node[above,pos=0.4] {$1$} (H-1);
				\draw[nnedge] (I-2) --node[above,pos=0.4] {$-1$} (H-1);
				\draw[nnedge] (I-2) --node[above,pos=0.4] {$1$} (H-2);
				\draw[nnedge] (I-3) --node[above,pos=0.4] {$-1$} (H-2);
				\draw[nnedge] (I-3) --node[above,pos=0.4] {$1$} (H-3);
				\draw[nnedge] (I-4) --node[above,pos=0.4] {$-1$} (H-3);
				
				\draw[nnedge] (H-1) --node[above,pos=0.4] {$5$} (O-1);
				\draw[nnedge] (H-2) --node[above,pos=0.4] {$-3$} (O-1);
				
				\draw[nnedge] (H-2) --node[above,pos=0.4] {$3$} (O-2);
				\draw[nnedge] (H-3) --node[above,pos=0.4] {$-5$} (O-2);
				
				\node[above=0.01cm of H-1] (b1) {$\relu{}$};
				\node[above=0.01cm of H-2] (b1) {$\relu{}$};
				\node[above=0.01cm of H-3] (b1) {$\relu{}$};
			\end{tikzpicture}
		\end{center}
		\caption{The initial network.}
		\label{fig:ex1}
	\end{subfigure}
	\begin{subfigure}[t]{0.5\linewidth}
		\def\layersep{1.2cm}
		\begin{center}
			\begin{tikzpicture}[shorten >=1pt,->,draw=black!50, node distance=\layersep,font=\footnotesize]
				
				\node[input neuron] (I-1) at (0,-1) {$x_1$};
				\node[input neuron] (I-2) at (0,-2) {$x_2$};
				\node[input neuron] (I-3) at (0,-3) {$x_3$};
				\node[input neuron] (I-4) at (0,-4) {$x_4$};
				
				\node[hidden neuron] (H-1) at (1.5*\layersep, -1.5) {$v_1$};
				\node[hidden neuron] (H-2) at (1.5*\layersep,-2.5) {$v_2$};
				\node[hidden neuron] (H-3) at (1.5*\layersep, -3.5) {$v_3$};
				
				\node[output neuron] (O-1) at (3*\layersep,-2) {$y_1$};
				\node[output neuron] (O-2) at (3*\layersep,-3) {$y_2$};
				
				\node[constructed neuron new] (C-1) at (4.5*\layersep,-2.5) {$y_\theta$};
				
				\draw[nnedge] (I-1) --node[above,pos=0.4] {$1$} (H-1);
				\draw[nnedge] (I-2) --node[above,pos=0.4] {$-1$} (H-1);
				\draw[nnedge] (I-2) --node[above,pos=0.4] {$1$} (H-2);
				\draw[nnedge] (I-3) --node[above,pos=0.4] {$-1$} (H-2);
				\draw[nnedge] (I-3) --node[above,pos=0.4] {$1$} (H-3);
				\draw[nnedge] (I-4) --node[above,pos=0.4] {$-1$} (H-3);
				
				\draw[nnedge] (H-1) --node[above,pos=0.4] {$5$} (O-1);
				\draw[nnedge] (H-2) --node[above,pos=0.4] {$-3$} (O-1);
				
				\draw[nnedge] (H-2) --node[above,pos=0.4] {$3$} (O-2);
				\draw[nnedge] (H-3) --node[above,pos=0.4] {$-5$} (O-2);
				
				\draw[nnedge] (O-1) --node[above,pos=0.4] {$1$} (C-1);
				\draw[nnedge] (O-2) --node[above,pos=0.4] {$-1$} (C-1);

				\node[above=0.01cm of H-1] (b1) {$\relu{}$};
				\node[above=0.01cm of H-2] (b1) {$\relu{}$};
				\node[above=0.01cm of H-3] (b1) {$\relu{}$};
			\end{tikzpicture}
		\end{center}
		\caption{Adding a construct for $y_1 \geq y_2$.}
		\label{fig:ex2}
	\end{subfigure}
	\begin{subfigure}[t]{\linewidth}
		\vspace{0.5cm}
		\def\layersep{1.2cm}
		\begin{center}
			\begin{tikzpicture}[shorten >=1pt,->,draw=black!50, node distance=\layersep,font=\footnotesize]
				
				\node[input neuron] (I-1) at (0,-1) {$x_1$};
				\node[input neuron] (I-2) at (0,-2) {$x_2$};
				\node[input neuron] (I-3) at (0,-3) {$x_3$};
				\node[input neuron] (I-4) at (0,-4) {$x_4$};
				
				\node[hidden neuron] (H-1) at (1.5*\layersep, -1.5) {$v_1$};
				\node[hidden neuron] (H-2) at (1.5*\layersep,-2.5) {$v_2$};
				\node[hidden neuron] (H-3) at (1.5*\layersep, -3.5) {$v_3$};
				
				\node[output neuron] (O-1) at (3*\layersep,-2) {$y_1$};
				\node[output neuron] (O-2) at (3*\layersep,-3) {$y_2$};
				
				\node[constructed neuron] (C-1) at (4.5*\layersep,-2) {$y_\theta$};
				
				\node[constructed neuron new] (C-2) at (4.5*\layersep,-3) {$y_\psi$};
				\node[constructed neuron new] (C-3) at (0,-5) {$\epsilon$};
				
				\node[constructed neuron] (C-4) at (1.5*\layersep,-5) {};
				\node[constructed neuron] (C-5) at (3*\layersep,-5) {};
				
				\draw[nnedge] (I-1) --node[above,pos=0.4] {$1$} (H-1);
				\draw[nnedge] (I-2) --node[above,pos=0.4] {$-1$} (H-1);
				\draw[nnedge] (I-2) --node[above,pos=0.4] {$1$} (H-2);
				\draw[nnedge] (I-3) --node[above,pos=0.4] {$-1$} (H-2);
				\draw[nnedge] (I-3) --node[above,pos=0.4] {$1$} (H-3);
				\draw[nnedge] (I-4) --node[above,pos=0.4] {$-1$} (H-3);
				
				\draw[nnedge] (H-1) --node[above,pos=0.4] {$5$} (O-1);
				\draw[nnedge] (H-2) --node[above,pos=0.4] {$-3$} (O-1);
				
				\draw[nnedge] (H-2) --node[above,pos=0.4] {$3$} (O-2);
				\draw[nnedge] (H-3) --node[above,pos=0.4] {$-5$} (O-2);
				
				\draw[nnedge] (O-1) --node[above,pos=0.4] {$1$} (C-1);
				\draw[nnedge] (O-2) --node[above,pos=0.4] {$-1$} (C-1);
				
				\draw[nnedge] (C-3) --node[above,pos=0.4] {$1$} (C-4);
				\draw[nnedge] (C-4) --node[above,pos=0.4] {$1$} (C-5);
				\draw[nnedge] (C-5) --node[above,pos=0.4] {$1$} (C-2);
				\draw[nnedge] (O-2) --node[above,pos=0.4] {$-1$} (C-2);
				
				\node[above=0.01cm of H-1] (b1) {$\relu{}$};
				\node[above=0.01cm of H-2] (b1) {$\relu{}$};
				\node[above=0.01cm of H-3] (b1) {$\relu{}$};
				
			\end{tikzpicture}
			\vspace{0.2cm}
			\caption{Adding a construct for $y_2 > 0$.}
			\vspace{0.2cm}
			\label{fig:ex3}
		\end{center}
	\end{subfigure}
	\begin{subfigure}[t]{\linewidth}
		\vspace{0.5cm}
		\def\layersep{1.2cm}
		\begin{center}
			\begin{tikzpicture}[shorten >=1pt,->,draw=black!50, node distance=\layersep,font=\footnotesize]
				
				\node[input neuron] (I-1) at (0,-1) {$x_1$};
				\node[input neuron] (I-2) at (0,-2) {$x_2$};
				\node[input neuron] (I-3) at (0,-3) {$x_3$};
				\node[input neuron] (I-4) at (0,-4) {$x_4$};
				
				\node[hidden neuron] (H-1) at (1.5*\layersep, -1.5) {$v_1$};
				\node[hidden neuron] (H-2) at (1.5*\layersep,-2.5) {$v_2$};
				\node[hidden neuron] (H-3) at (1.5*\layersep, -3.5) {$v_3$};
				
				\node[output neuron] (O-1) at (3*\layersep,-2) {$y_1$};
				\node[output neuron] (O-2) at (3*\layersep,-3) {$y_2$};
				
				\node[constructed neuron] (C-1) at (4.5*\layersep,-2) {$y_\theta$};
				
				\node[constructed neuron] (C-2) at (4.5*\layersep,-3) {$y_\psi$};
				\node[constructed neuron] (C-3) at (0,-5) {$\epsilon$};
				
				\node[constructed neuron] (C-4) at (1.5*\layersep,-5) {};
				\node[constructed neuron] (C-5) at (3*\layersep,-5) {};
				
				\node[constructed neuron] (C-6) at (6*\layersep,-2) {};
				\node[constructed neuron] (C-7) at (6*\layersep,-3) {};
				\node[constructed neuron new] (C-8) at (7.5*\layersep,-2.5) {$y_\varphi$};

				\draw[nnedge] (I-1) --node[above,pos=0.4] {$1$} (H-1);
				\draw[nnedge] (I-2) --node[above,pos=0.4] {$-1$} (H-1);
				\draw[nnedge] (I-2) --node[above,pos=0.4] {$1$} (H-2);
				\draw[nnedge] (I-3) --node[above,pos=0.4] {$-1$} (H-2);
				\draw[nnedge] (I-3) --node[above,pos=0.4] {$1$} (H-3);
				\draw[nnedge] (I-4) --node[above,pos=0.4] {$-1$} (H-3);
				
				\draw[nnedge] (H-1) --node[above,pos=0.4] {$5$} (O-1);
				\draw[nnedge] (H-2) --node[above,pos=0.4] {$-3$} (O-1);
				
				\draw[nnedge] (H-2) --node[above,pos=0.4] {$3$} (O-2);
				\draw[nnedge] (H-3) --node[above,pos=0.4] {$-5$} (O-2);
				
				\draw[nnedge] (O-1) --node[above,pos=0.4] {$1$} (C-1);
				\draw[nnedge] (O-2) --node[above,pos=0.4] {$-1$} (C-1);
				
				\draw[nnedge] (C-3) --node[above,pos=0.4] {$1$} (C-4);
				\draw[nnedge] (C-4) --node[above,pos=0.4] {$1$} (C-5);
				\draw[nnedge] (C-5) --node[above,pos=0.4] {$1$} (C-2);
				\draw[nnedge] (O-2) --node[above,pos=0.4] {$-1$} (C-2);
				
				\draw[nnedge] (C-1) --node[above,pos=0.4] {$-1$} (C-6);
				\draw[nnedge] (C-2) --node[above,pos=0.4] {$-1$} (C-7);
				\draw[nnedge] (C-6) --node[above,pos=0.4] {$-1$} (C-8);
				\draw[nnedge] (C-7) --node[above,pos=0.4] {$-1$} (C-8);
				
				\node[above=0.01cm of H-1] (b1) {$\relu{}$};
				\node[above=0.01cm of H-2] (b1) {$\relu{}$};
				\node[above=0.01cm of H-3] (b1) {$\relu{}$};
				\node[above=0.01cm of C-6] (b1) {$\relu{}$};
				\node[above=0.01cm of C-7] (b1) {$\relu{}$};
				\node[above=0.01cm of C-8] (b1) {$\relu{}$};
			\end{tikzpicture}
			\vspace{0.2cm}
			\caption{Adding a construct for $\varphi$, as a conjunction.}
			\label{fig:ex4}
		\end{center}
	\end{subfigure}
	\vspace{0.5cm}
	\caption{Construction of a reachability problem for $\nn \models \varphi$. }
	\label{fig:example}
\end{figure*}
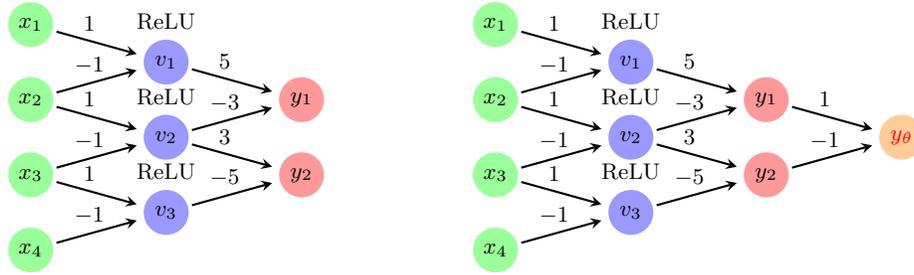
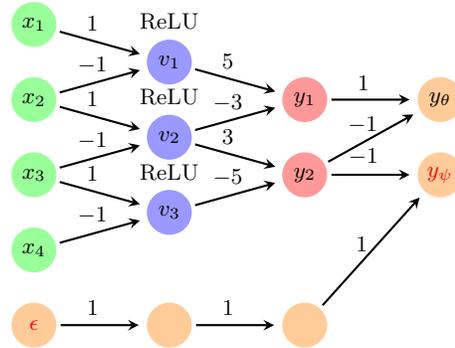
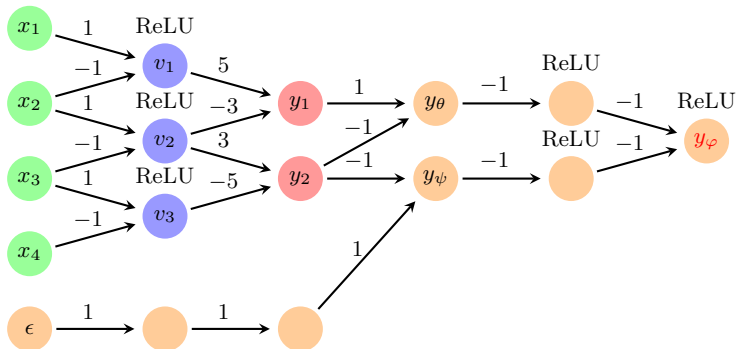

We now prove the theorem by induction on the generating sequence of $\varphi$;
that is, a sequence of sub-formulas of $\varphi$: $\varphi_1,...,\varphi_n$ such that
$\forall i,j$ if $j>i$ then $\varphi_j$ cannot be a sub-formula of
$\varphi_i$, and $\varphi_n = \varphi$. This allows inductive proofs
over the formulas~\cite{Sh67}. For example, a generating
sequence for the formula
\begin{center}
	$\varphi := \exists x,y.(3x \geq 7)\wedge\lnot(y \geq x)$
\end{center} 
is: 
\begin{center}
	$y\geq x, 3x \geq 7, \lnot(y \geq x),(3x \geq 7)\wedge\lnot(y \geq x),\exists x,y.(3x \geq 7)\wedge\lnot(y \geq x) $
\end{center}
\begin{proof}
  Without loss of generality, assume that $\varphi$ is composed of
  atoms, negations, and conjunctions. In addition, assume that each
  variable $y_j$ is an output variable (otherwise, we may add neurons
  with the identity as activation function from the neuron outputting
  $y_j$ to the output layer). For every step $i$ in the generating
  sequence $\varphi_1,...,\varphi_k = \varphi$, we add a constant
  number of output neurons, such that for any $x \in \rn{m}$,
  the resulting DNN $\nn_i'$, satisfies $\nn_i' \geq 0$ (for the last constructed output neuron) if and only if $\nn(x) \models \varphi_i$. Below we
  explain the construction and prove its
  correctness; and in Figure~\ref{fig:gadgets} we show its visual
  representation.

  \medskip
  \noindent
    \textbf{Base Cases:} 
\begin{enumerate}
\item Let $\varphi := \top $. In this case, $\nn'$ is constructed from
  $\nn$ by adding a single affine neuron with no activation function, and with its input edges with
  weight $0$ from all output nodes of $\nn$. This also maintains the convention that the
  output layer does not have an activation function. Therefore,
  $\forall x\in \rn{m}: \nn'(x,\epsilon) \geq 0$ if and only if
  $\underset{i}{\sum} 0 \geq 0 $, which is equivalent to
  $\top$. The case of $\bot$ is covered by our handling of negations.
  
\item Let $\varphi := \underset{i}{\sum} c_i \cdot y_i + b \geq 0 $.
  In this case, $\nn'$ is constructed from $\nn$ by adding a single
  affine neuron with no activation function, with its input edges with
  weight $c_i$ from every output neuron $y_i$ of $\nn$, and a bias
  $b$. Therefore, $\forall x \in \rn{m}$ and $\nn(x) = y $ we have
  that $\underset{i}{\sum} c_i \cdot y_i + b \geq 0 $ if and only if
  $\nn'(x,\epsilon) \geq 0$.  We note that equality can be handled using
  conjunctions, while strict inequalities can be handled using
  negations.
\end{enumerate}
\medskip
\textbf{Inductive step:} 
\begin{enumerate}
	\item Let $ \varphi := \psi \land \theta$, and let $y_\psi,
          y_\theta$ be the values of the neurons such that $y_\psi \geq 0, y_\theta \geq 0$ if and only if $\nn(x) \models \psi, \theta$, respectively. Consider:
\[
  y_\varphi = -\relu(-y_\psi) -\relu(-y_\theta)
\]
	In this case, we have that $y_\varphi \geq 0$ if and only if $y_\psi \geq 0 \land y_\theta \geq 0$. We can see this since if $y_\psi \geq 0 \land y_\theta \geq 0$ then both $-\relu(-y_\psi)$ and $-\relu(-y_\theta)$ equal zero. Otherwise, at least one of $-\relu(-y_\psi)$ and $-\relu(-y_\theta)$ is negative (and the other is non-positive). 
	Thus, we add two \relu{} neurons, with a single $-1$ input edge from each of the nodes corresponding to $y_\psi, y_\theta$, respectively. We then add a third neuron with two $-1$ edges from the \relu{} nodes and no activation function. 
	
      \item Let $\varphi := \lnot\psi$, and let $y_\psi$ be the value
        of the neuron such that $y_\psi \geq 0$ if and only if
        $\nn(x) \models \psi$. In this case, we first need to add a
        new $\epsilon_\varphi$ input neuron, and restrict it to
        $\epsilon_\varphi > 0$. Then, observe that
        $\lnot (y_\psi \geq 0) \equiv y_\psi < 0 $ if and only if
        there exists some $\epsilon > 0$ s.t.
        $ \epsilon + y_\psi \leq 0$, or equivalently
        $-\epsilon - y_\psi \geq 0$. Therefore, we add a new neuron
        with no activation function, with a skip connection from the
        $\epsilon_\varphi$ neuron with weight $-1$, and with a $-1$ weight from
        $y_\psi$, resulting in $y_\varphi = -\epsilon_\varphi - y_\psi$.
        We note that since there are finitely many
        such constructions, we can choose the minimal $\epsilon$
        implied by 
        all of them, and again choose the minimum of it and $1$. Thus,  a
        single $\epsilon \in \left( 0,1 \right]$ suffices.  In
        addition, the use of the skip connections can be replaced with
        a line of \relu{} neurons, starting with the $\epsilon$ neuron and
        feed-forwarding to a neuron on every layer. This construction
        does not affect the asymptotic size of $\nn'$.
\end{enumerate}

On every step of the recursion, we added a constant number of
neurons to the network, such that $\forall x \in \rn{m}$, $\nn'(x,\epsilon) \geq 0$ if and
only if $\nn(x) \models \varphi_i$. This concludes our proof.
\end{proof}

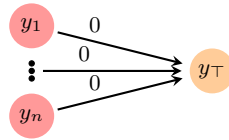
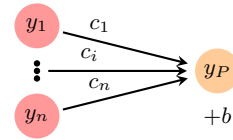
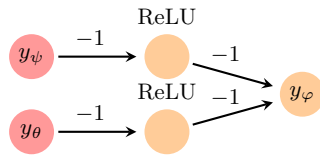
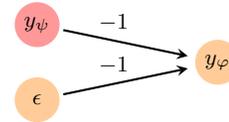
\begin{figure*}[h!]
	\begin{subfigure}[t]{0.45\linewidth}
		\def\layersep{1.2cm}
		\begin{center}
			\begin{tikzpicture}[shorten >=1pt,->,draw=black!50, node distance=\layersep,font=\footnotesize]
				
				\node[output neuron] (O-1) at (0,0) {$y_1$};
				\filldraw [black] (0,-0.5) circle (1pt);
				\filldraw [black] (0,-0.6) circle (1pt);
				\filldraw [black] (0,-0.7) circle (1pt);
				\node[output neuron] (O-2) at (0,-1.2) {$y_n$};
				
				\node[constructed neuron] (C-1) at (2*\layersep,-0.6) {$y_\top$};
				
				\draw[nnedge] (O-1) --node[above,pos=0.3] {$0$} (C-1);
				\draw[nnedge] (0.1,-0.6) --node[above,pos=0.3] {$0$} (C-1);
				\draw[nnedge] (O-2) --node[above,pos=0.3] {$0$} (C-1);

			\end{tikzpicture}
		\end{center}
		\caption{ An atom predicate of the form $\top$.}
	\end{subfigure}
	\begin{subfigure}[t]{0.55\linewidth}
		\def\layersep{1.2cm}
		\begin{center}
			\begin{tikzpicture}[shorten >=1pt,->,draw=black!50, node distance=\layersep,font=\footnotesize]
				
				\node[output neuron] (O-1) at (0,0) {$y_1$};
				\filldraw [black] (0,-0.5) circle (1pt);
				\filldraw [black] (0,-0.6) circle (1pt);
				\filldraw [black] (0,-0.7) circle (1pt);
				\node[output neuron] (O-2) at (0,-1.2) {$y_n$};
				
				\node[constructed neuron] (C-1) at (2*\layersep,-0.6) {$y_P$};
				
				\draw[nnedge] (O-1) --node[above,pos=0.3] {$c_1$} (C-1);
				\draw[nnedge] (0.1,-0.6) --node[above,pos=0.3] {$c_i$} (C-1);
				\draw[nnedge] (O-2) --node[above,pos=0.3] {$c_n$} (C-1);
				
				\node[below=0.05cm of C-1] (b1) {$+b$};
				
			\end{tikzpicture}
		\end{center}
		\caption{An atom predicate of the form $P:=\underset{i}{\sum} c_i \cdot y_i + b \geq 0$.}
	\end{subfigure}
	\begin{subfigure}[t]{0.45\linewidth}
		\def\layersep{1.2cm}
		\begin{center}
			\begin{tikzpicture}[shorten >=1pt,->,draw=black!50, node distance=\layersep,font=\footnotesize]
				
				\node[output neuron] (O-1) at (0,0) {$y_\psi$};
				\node[output neuron] (O-2) at (0,-1) {$y_\theta$};
				
				\node[constructed neuron] (C-1) at (1.5*\layersep,0) {};
				\node[constructed neuron] (C-2) at (1.5*\layersep,-1) {};
				\node[constructed neuron] (C-3) at (3*\layersep,-0.5) {$y_\varphi$};
				
				\draw[nnedge] (O-1) --node[above,pos=0.4] {$-1$} (C-1);
				\draw[nnedge] (O-2) --node[above,pos=0.4] {$-1$} (C-2);
				\draw[nnedge] (C-1) --node[above,pos=0.4] {$-1$} (C-3);
				\draw[nnedge] (C-2) --node[above,pos=0.4] {$-1$} (C-3);
				
				\node[above=0.01cm of C-1] (b1) {$\relu{}$};
				\node[above=0.01cm of C-2] (b1) {$\relu{}$};

			\end{tikzpicture}
		\end{center}
		\caption{ A formula of the form $\varphi =  \psi\land \theta $.}
	\end{subfigure}
	\begin{subfigure}[t]{0.55\linewidth}
		\def\layersep{1.2cm}
		\begin{center}
			\begin{tikzpicture}[shorten >=1pt,->,draw=black!50, node distance=\layersep,font=\footnotesize]
				
				\node[output neuron] (O-1) at (0,0) {$y_\psi$};
				
				\node[constructed neuron] (C-1) at (0,-1) {$\epsilon$};
				\node[constructed neuron] (C-2) at (2*\layersep,-0.5) {$y_\varphi$};
				
				\draw[nnedge] (O-1) --node[above,pos=0.4] {$-1$} (C-2);
				\draw[nnedge] (C-1) --node[above,pos=0.4] {$-1$} (C-2);
				
			\end{tikzpicture}
		\end{center}
		\caption{ A formula of the form $\varphi = \lnot \psi$. }
	\end{subfigure}
	\caption{Constructs for each step of the induction.}
	\label{fig:gadgets}
	\vspace{-0.2cm}
\end{figure*}

\section{Related Work}
\label{sec:Related}
The complexity of DNN verification has been studied mainly for DNNs
with piecewise-linear activation functions --- specifically, the 
\relu{} function. It has previously been shown that DNN verification is
NP-complete, even for simple
specifications~\cite{KaBaDiJuKo21,SaLa21}. However, when certain
restricted classes of DNN architectures and specifications are considered, DNNs with \relu{}s
only can be verified in polynomial time~\cite{FeSh21}.  The
verification complexity and computability in the case of reactive
systems controlled by \relu{} DNNs (i.e., the DNN acts as an agent
that repeatedly interacts with an environment) has also
been studied recently~\cite{AkKeLoPi19,AkLoMaPi18}. One
recent work showed that verifying CTL properties in this context is
undecidable~\cite{AkBoKoLo22}. In our work, however, we consider DNNs as stand-alone
functions.

\sloppy When considering realistic implementations of the \relu{}
function, e.g., in \emph{quantized neural
  networks}~\cite{HeLeZi21}, the DNN verification problem for
bit-vector specifications is PSPACE-hard, introducing a big complexity
gap from the case of ideal mathematical form. When specific models
of \emph{graph neural networks}~\cite{ZhCuHuZhYaLiWaLiSu20} are
considered, the verification problem is undecidable~\cite{SaLa23}.

For DNNs with Sigmoidal activation functions, two main
results are presently known. First, it was shown that reachability
analysis with some error tolerance $\epsilon$, for any
Lipschitz-continuous activation function, is NP-complete in the size
of the network and of $\epsilon$~\cite{RuHuZiKw18}. Second, it was
shown that the decidability of $\mathcal{T}_{\exp}$ implies the
decidability of verifying DNNs with Sigmoidal activation functions, and
that verifying such DNNs with a single hidden layer is
decidable~\cite{IvWeAlPaLe19}.
Our work here is another step towards a
better understanding of the complexity of verifying such DNNs.
The computational power of \emph{Recurrent} Neural Networks with Sigmoidal activation functions has been studied as well, with Turing completeness results for Sigmoidal RNNs~\cite{Ca19}. This can be further used to study the verification of Sigmoidal RNNs.

The complexity of formal analysis of DNNs with other functions, such as \emph{Gaussian} and \emph{arctan} has also been studied, showing the verification problem is at least as hard as deciding formulas in $\mathcal{T}_{\mathbb{R}}$~\cite{Wu23}.

The connections between DNN verification and DNN reachability have
also been studied before. Most prominently, it was shown that any
local-robustness verification query can be reduced to a DNN
reachability query~\cite{RuHuZiKw18}. In addition, a similar construction showing the equivalence between verification and reachability has been used before~\cite{ElGoKa20}, though for a specific example without a formal proof.

\section{Conclusion and Future Work}
\label{sec:Conclusion}
Our results show that for DNNs with \relu{}, $\sigma$, $\tanh$ and \text{NLReLU}
activation functions, the decidability of the verification problem is
equivalent to a well-known open problem; and that it can be reduced to a
problem whose approximation is decidable, and for whose complexity an
upper bound is known. This was achieved by reducing the verification problem to
the corresponding reachability problem. These results show a significant difference between
the verification problem for DNNs with piecewise-smooth activation functions and
for DNNs with piecewise-linear activation 
functions, which is known to be 
NP-complete~\cite{KaBaDiJuKo21,SaLa21}.

Moving forward, one goal that we plan to pursue is a version of our
first result that does not rely on the $\text{NLReLU}$ function, which
is not as mainstream as the other functions that we
considered. Although discarding this function does not alter the first
direction of the proof, the second reduction currently requires it;
and we plan to circumvent this requirement by defining a reduction
from a $\Sigma_{\exp}$-formula to a formula in the signature of
$\mathcal{T}_\mathbb{R} \cup
Th(\mathbb{R},+,-,{\cdot}_{q\in\mathbb{Q}},0,1,<,\sigma,\tanh)$, and
then using a combination of the decision procedures for these two
theories. It is noteworthy that the Nelson-Oppen method~\cite{NeOp79}
cannot be directly applied here, since it requires the combined theories to be
disjoint, which is not the case. Several generalizations of the
Nelson-Oppen method for non-disjoint theories have previously been
proposed~\cite{TiRi03,Gh04,NiRiRu09,Ri96}, and we speculate that these
could be useful in this context.

Another interesting direction that we plan to pursue is to combine our
work with approaches for switching between different kinds of machine
learning models. For example, it would be intriguing to study whether
DNNs with smooth activation functions can be reduced to decision
trees or to neural networks with a fixed number of layers, as can apparently be done for piecewise-linear
DNNs~\cite{Ay22,ViSc23}. Equivalently, fundamental differences between piecewise-linear
DNNs and smooth DNNs might imply similar differences between other
classes of machine learning models.

Our second result could also be generalized, in two different
manners. First, our construction could be applied to verification
queries that involve multiple DNNs, e.g., verification queries used
for proving DNN equivalence~\cite{NaKaRySaWa17}. This is true since
for two DNNs $\nn_1, \nn_2$ operating on the same domain, the
verification query $\nn_1(x) \overset{?}{=}\nn_2(x)$ can be reduced to
$\nn'(x) \overset{?}{=} 0$, where $\nn'$ is constructed from copies of
$\nn_1,\nn_2$ with outputs $y_1, y_2$, and where additional neurons
are used to stipulate that $y_3 \geq 0 \iff y_1 = y_2$, using our
construction. In this case, we have that $\nn' = O(|\nn_1|
+|\nn_2|)$. An illustration of this construction appears
in Figure~\ref{fig:equiisreach}.
These results, in turn, could be generalized to queries that involve any finite number of DNNs.

\begin{figure}[h]
	\def\layersep{1.2cm}
	\begin{center}
		\begin{tikzpicture}[shorten >=1pt,->,draw=black!50, node distance=\layersep,font=\footnotesize]
			
			\node[input neuron] (I-1) at (0,-1) {$x_1$};
			\node[input neuron] (I-2) at (0,-2) {$x_2$};
			\node[input neuron] (I-3) at (0,-3) {$x_3$};
			
			\node[hidden neuron] (H-1) at (1.5*\layersep, -1.5) {$v_1$};
			\node[hidden neuron] (H-2) at (1.5*\layersep,-2.5) {$v_2$};
			
			\node[output neuron] (O-1) at (3*\layersep,-2) {$y_1$};

			\draw[nnedge] (I-1) --node[above,pos=0.4] {$1$} (H-1);
			\draw[nnedge] (I-2) --node[above,pos=0.4] {$-1$} (H-1);
			\draw[nnedge] (I-2) --node[above,pos=0.4] {$1$} (H-2);
			\draw[nnedge] (I-3) --node[above,pos=0.4] {$-1$} (H-2);
			
			\draw[nnedge] (H-1) --node[above,pos=0.4] {$1$} (O-1);
			\draw[nnedge] (H-2) --node[above,pos=0.4] {$-1$} (O-1);

			\node[above=0.01cm of H-1] (b1) {$\relu{}$};
			\node[above=0.01cm of H-2] (b1) {$\relu{}$};

			\node[input neuron] (I2-1) at (0,-4) {$x_5$};
			\node[input neuron] (I2-2) at (0,-5) {$x_6$};
			\node[input neuron] (I2-3) at (0,-6) {$x_7$};
			\node[input neuron] (I2-4) at (0,-7) {$x_8$};
			
			\node[hidden neuron] (H2-1) at (1.5*\layersep, -4.5) {$u_1$};
			\node[hidden neuron] (H2-2) at (1.5*\layersep,-5.5) {$u_2$};
			\node[hidden neuron] (H2-3) at (1.5*\layersep, -6.5) {$u_3$};
			
			\node[output neuron] (O2-1) at (3*\layersep,-5.5) {$y_2$};

			\draw[nnedge] (I2-1) --node[above,pos=0.4] {$1$} (H2-1);
			\draw[nnedge] (I2-2) --node[above,pos=0.4] {$-5$} (H2-1);
			\draw[nnedge] (I2-2) --node[above,pos=0.4] {$1$} (H2-2);
			\draw[nnedge] (I2-3) --node[above,pos=0.4] {$-5$} (H2-2);
			\draw[nnedge] (I2-3) --node[above,pos=0.4] {$1$} (H2-3);
			\draw[nnedge] (I2-4) --node[above,pos=0.4] {$-5$} (H2-3);
			
			\draw[nnedge] (H2-1) --node[above,pos=0.4] {$4$} (O2-1);
			\draw[nnedge] (H2-2) --node[above,pos=0.4] {$4$} (O2-1);
			\draw[nnedge] (H2-3) --node[above,pos=0.4] {$4$} (O2-1);
			
			\node[above=0.01cm of H2-1] (b1) {$\relu{}$};
			\node[above=0.01cm of H2-2] (b1) {$\relu{}$};
			\node[above=0.01cm of H2-3] (b1) {$\relu{}$};
			
			\node[constructed neuron] (C-1) at (4.5*\layersep,-2.5) {$y_{y_1\leq y_2}$};
			\node[constructed neuron] (C-2) at (4.5*\layersep,-5.5) {$y_{y_1\geq y_2}$};
			\node[constructed neuron] (C-3) at (6*\layersep,-2.5) {};
			\node[constructed neuron] (C-4) at (6*\layersep,-5.5) {};
			\node[constructed neuron] (C-5) at (7.5*\layersep,-4) {$y_3$};
			
			\draw[nnedge] (O-1) --node[above,pos=0.4] {$-1$} (C-1);
			\draw[nnedge] (O2-1) --node[above,pos=0.4] {$1$} (C-1);
			
			\draw[nnedge] (O-1) --node[above,pos=0.4] {$1$} (C-2);
			\draw[nnedge] (O2-1) --node[below,pos=0.4] {$-1$} (C-2);
			
			\draw[nnedge] (C-2) --node[below,pos=0.4] {$-1$} (C-4);
			\draw[nnedge] (C-1) --node[above,pos=0.4] {$-1$} (C-3);
			
			\draw[nnedge] (C-3) --node[above,pos=0.5] {$-1$} (C-5);
			\draw[nnedge] (C-4) --node[below,pos=0.5] {$-1$} (C-5);
			
			\node[above=0.01cm of C-1] (b1) {$\relu{}$};
			\node[below=0.01cm of C-2] (b1) {$\relu{}$};
			\node[above=0.01cm of C-3] (b1) {$\relu{}$};
			\node[=0.01cm of C-4] (b1) {$\relu{}$};
			
		\end{tikzpicture}
		\caption{Reducing DNN equivalence to DNN reachability.}
		\label{fig:equiisreach}
	\end{center}
\end{figure}
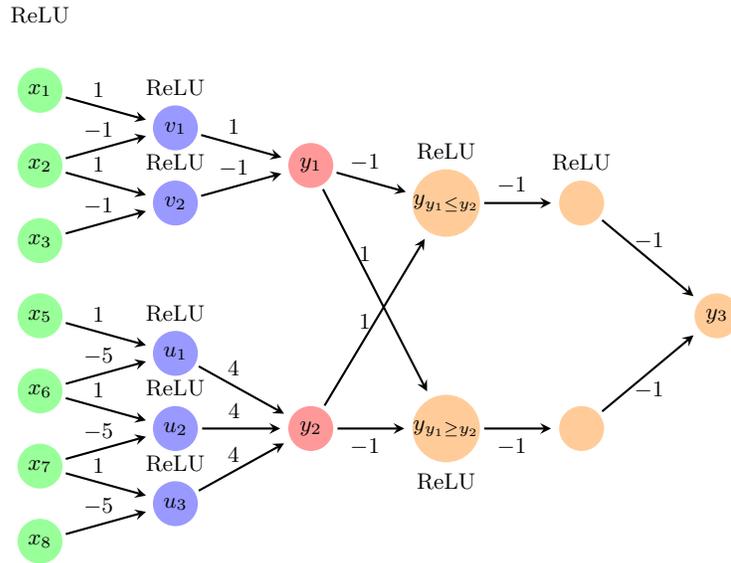

Second, similar constructions can support specification formulas over
arithmetics that include any activation function (even piecewise-smooth). In
this case, the number of added neurons is proportional not only to the
sizes of the original network and the formula, but also to the number
of activation functions composing the atoms. This construction is
straightforward, and is omitted.

Our second result provides a notion of estimation for the DNN verification problem in general. That is, we could relax any DNN verification query to its equivalent $\epsilon$-tolerant reachability query. This effectively allows
an error tolerance in the value of the output neuron of the resulting
network. However, the intuitive definition for approximating DNN verification
is to introduce error tolerance to the values of all neurons. To that end, we plan to investigate the connections between these two definitions of relaxation, and the advantages of using each one.

A final line of work that we intend to pursue in the future is to
consider a more realistic framework of verification, with a concrete
implementation of $\sigma$, rather than its pure mathematical
form. This is similar to what was done for DNNs with \relu{}
activation functions~\cite{HeLeZi21}. In addition, we intend to
characterize \emph{decidable fragments} of the DNN verification
problem, by restricting specifications and/or architectures; that is,
we plan to identify sufficient conditions on the DNNs and
specifications, which would render the resulting verification problem
decidable. For such decidable fragments, studying the computational
complexity of the verification problem is yet another intriguing line
of work. Similar research was conducted in the context of differential privacy~\cite{BaChKrSiVi21}, and it is interesting to study whether the decidable fragments identified in this research could be useful for DNN verification as well. We also intend to further explore implications of the \emph{Quasi-Decidability} of $\mathcal{T}_{exp}$~\cite{FrRaZg11} on DNN verification.

\mysubsection{Acknowledgments} This work was supported by the Israel Science 
	Foundation (grant number 619/21), the Binational Science Foundation (grant numbers
	2020250, 2021769, 
	2020704), 
	and by the National Science Foundation (grant numbers 1814369 
	and 2110397). 

\bibliographystyle{abbrv}
{\footnotesize
	\bibliography{decidability}}
\end{document}